\documentclass[letterpaper, 10 pt, conference]{ieeeconf}  

\IEEEoverridecommandlockouts                              
\usepackage{caption}\usepackage{dsfont}
\usepackage{subfigure}
\usepackage{hyperref}
\usepackage{amsmath,amssymb,amsfonts}
\usepackage{bm,bbm} 
\usepackage{graphicx,color}
\newtheorem{theorem}{Theorem}
\newtheorem{definition}{Definition}
\newtheorem{proposition}{Proposition}
\newtheorem{lemma}{Lemma}
\newtheorem{corollary}{Corollary}
\newtheorem{example}{Example}

\newcommand{\ba}{\begin{array}}
\newcommand{\ea}{\end{array}}

\newcommand{\be}{\begin{equation}}
\newcommand{\ee}{\end{equation}}

\newcommand{\ds}{\displaystyle}

\newcommand{\eps}{\varepsilon}

\newcommand{\mc}{\mathcal}

\newcommand{\ov}{\overline}

\newcommand{\1}{\mathds{1}}

\newcommand{\R}{\mathbb{R}}

\newcommand{\prodd}{\prod\limits}

\newcommand{\se}{\text{ if }}

\renewcommand{\span}{\text{span}}
\def\qed{\hfill \vrule height 7pt width 7pt depth 0pt \medskip}

\title{\LARGE \bf
Robustness of Nash Equilibria in Network Games
}

\author{Laura Arditti, Giacomo Como, Fabio Fagnani, and Martina Vanelli
\thanks{This work was partially supported by MIUR grant Dipartimenti di Eccellenza 2018--2022 [CUP: E11G18000350001], the Swedish Research Council [2015-04066], and the Compagnia di San Paolo.}
\thanks{The authors are with the Department of Mathematical Sciences, Politecnico di Torino, Corso Duca degli Abruzzi 24, 10129, Torino, Italy. 
        Email: {\tt\small laura.arditti@polito.it}, {\tt\small giacomo.como@polito.it}, {\tt\small fabio.fagnani@polito.it}, {\tt\small martina.vanelli@polito.it}.}
        \thanks{The second author is also with the Department of Automatic Control, Lund University, Sweden.}%
}

\begin{document}

\maketitle
\thispagestyle{empty}
\pagestyle{empty}

\begin{abstract}

We analyze the robustness of (pure strategy) Nash equilibria for network games against perturbations of the players' utility functions. We first derive a simple characterization of the margin of robustness, defined as  the minimum magnitude of a perturbation that makes a Nash equilibrium of the original game stop being so in the perturbed game. Then, we investigate what the maximally robust equilibria are in some standard network games such as the coordination and the anti-coordination game. Finally, as an application, we provide some sufficient conditions for the existence of Nash equilibria in network games with a mixture of coordinating and anti-coordinating games.
\end{abstract}

\section{Introduction}\label{sec:introduction}
Robustness of Nash equilibria in game theory is typically addressed by introducing more refined equilibria concepts that can be proven to be stable with respect to certain families of perturbations (e.g. noisy payoffs, incomplete information games) \cite{KM,Tui}. In particular, robustness is considered as a quality that either is present or not for a certain game and a specific equilibrium.

Our study goes into a quite different direction and gives a twofold contribution. 

First, we propose a way of measuring robustness of a Nash equilibrium for any strategic form game that is much in the spirit of robust control theory as it represents the minimum energy that a perturbation of payoffs must possess in order to break the equilibrium. This quantity can then be dually characterized as the result of a straightforward optimization problem involving payoffs. Remarkably, for the family of potential games, the most robust equilibria do not necessarily coincide with the maxima of the potential. This is in sharp contrast with the results in \cite{Tui} where robustness is studied in the context of games with incomplete information.

Second, we consider games where players are split into two classes and we use this notion of robustness to investigate the structure of Nash equilibria of the sub-game obtained by freezing players in one of the two classes to some prescribed value. In particular, we study conditions under which such Nash equilibrium is independent from the freezing value. Finally, we apply this machinery to find explicit examples of Nash equilibria for games over networks consisting of a mixture of coordination and anti-coordination games. Such games have attracted much attention lately \cite{Ramazi.Riehl.Cao:2016} and, as they are not potential, the existence of Nash equilibria is not guaranteed.

We now briefly outline the content of this paper. Section \ref{sec:background} introduces all necessary notions of game theory. Section \ref{sec:margin-robustness} introduces the fundamental notion of robustness, it proposes a simple formula to compute it and present a number of examples. Section \ref{sec:robustness-subgames} studies robustness issues for Nash equilibria of sub-games of a given game and final Section \ref{sec:robustness-networkgames} focuses on network games and presents an application to study Nash equilibria for a network game consisting of a mixture of coordination and anti-coordination players.

\section{Background}\label{sec:background}
Throughout the paper, we shall consider strategic form games with finite nonempty player set $\mc V$ and finite nonempty action set $\mc A_i$ for each player $i$ in $\mc V$.  We shall denote by $\mathcal{X}= \prod_{i \in \mc V} \mc A_i$ the space of all players' strategy profiles and, for every player $i$ in $\mc V$, let $\mathcal{X}_{-i}= \prod_{j \in \mc V\setminus\{i\}} \mc A_j$ be the set of strategy profiles of all players except for player $i$. 
As customary, for a strategy profile $x$ in $\mc X$, the strategy profile of all players except for $i$ is denoted by $x_{-i}$ in $\mc X_{-i}$.
We shall refer to two strategy profiles $x$ and $y$  in $\mc X$ as  $i$-comparable and write $x\sim_i y$ when $x_{-i}=y_{-i}$, i.e., when $x$ and $y$ coincide except for possibly in their $i$-th entry.  
Let each player $i$ in $\mc V$ be equipped with a utility function $u_i:\mathcal{X}\to \mathbb{R}\,.$ We shall identify a game with player set $\mc V$ and strategy profile space $\mc X$ with the vector $u$ assembling all the players' utilities. 
The set of all games with player set $\mc V$ and strategy profile space $\mc X$, to be denoted by $\mc U$,  is isomorphic to the vector space $\R^{\mc V\times\mc X}$: we shall equip it with the infinity norm 
$$||u_i||_{\infty}=\max_{x\in\mc X}|u_i(x)|\,,\qquad i\in\mc V\,,$$
$$||u||_{\infty}=\max_{i\in\mc V}||u_i||_{\infty}\,.$$
For a game $u$ in $\mc U$, define 
\be\label{chi-def}\chi^{u}_i(x)=\min_{\substack{y\sim_i x\\ y\ne x}}\{u_i(x)-u_i(y)\}\,,\ee
for every player $i$ in $\mc V$ and strategy profile $x$ in $\mc X$. 

\begin{definition}\label{definition:margin}
A (pure strategy) Nash equilibrium for a game $u$ in $\mc U$ is a strategy profile $x^*$ in $\mc X$ such 
$\chi^{u}_i(x^*)\ge0$,
for every player $i$ in $\mc V$.
\end{definition}

While not all games have Nash equilibria, some important classes of games are known that always admit  Nash equilibria. One of them is the class of potential games \cite{Monderer.Shapley:1996}. 

\begin{definition} A game $u$ in $\mc U$ is said to be \emph{potential}  if there exists a potential function $\phi:\mathcal{X}\to\mathbb{R}$ such that 
	\begin{equation} \label{eq:potential}
	u_i(x)-u_i(y)=\phi(x)-\phi(y)\,,
	\end{equation}
	for every player $i$ in $\mc V$ and every pair of $i$-comparable strategy profiles $x\sim_i y$ in $\mc X$. 
\end{definition}

Throughout the paper, we will consider directed graphs $\mc G=(\mc V,\mc E)$ defined as the pair of a finite set of nodes $\mc E$ and a finite set of directed links $\mc E\subseteq\mc V\times\mc V$ and we shall denote by $\mc N_{i}=\{j\in\mc V:\,(i,j)\in\mc E\}$ the out-neighborhood of a node $i$ in $\mc G$. We will consider games on graphs according to the following definition, which was first introduced by Kearns \cite{Kearns.ea:2001}.

\begin{definition} A game $u$ in $\mc U$ is said to be \emph{graphical} on a graph $\mc G=(\mc V,\mc E)$ if the utility of each player $i$ in $\mc V$ depends only on her own action and on the actions of fellow players in her neighborhood in $\mc G$, i.e., if 
	\be\label{eq:graphical-def}u_i(x)=u_i(y)\,,\qquad\forall x,y\in\mc X\text{ s.t.~} x_{\mc N_{i}\cup\{i\}}=y_{\mc N_{i}\cup \{i\}}\,.\ee

\end{definition}
A relevant class of graphical games is that of pairwise-separable games, obtained by combining pairwise interactions according to the following definition.
\begin{definition}\label{definition:pairwise-network-game} A $\mc G$-game $u$ in $\mc U$ is said to be \emph{pairwise-separable}  if the utility of player $i$ in $\mc V$ is in the form 
	\begin{equation}\label{pairwise-game}
	u_i(x) = \sum_{j \in \mc N_{i}} u_{ij}(x_i,x_j) \qquad \forall x \in \mathcal{X}\,,
	\end{equation}
	where $u_{ij}:\mc A_i\times\mc A_j\to\R$ for $(i,j)$ in $\mc E$. 

\end{definition}

\begin{example}[Network coordination and anti-coordination]\label{example:coordination-anticoordination}
For an undirected weighted graph $\mc G = (\mc V, \mc E, W)$, and a binary vector $\xi$ in $\{-1,+1\}^{\mc V}$, let $u$ be a game with player set $\mc V$, binary action set $\mc A_i=\{-1,+1\}$ for every player $i$ in $\mc V$
and utilities
	\be\label{coord-anticoord-def}
	u_i(x) = \xi_i\sum_{j \in \mc V} W_{ij} x_ix_j\,,\qquad \forall x\in\mc X\,.
	\ee
For $\xi=+\1$, this is known as the pure (homogeneous) \emph{network coordination} game whereas for $\xi=-\1$ this the pure (homogeneous) \emph{network anti-coordination} game. 
Notice that by defining $u_{ij}(x_i,x_j) = \xi_i W_{ij} x_ix_j$ it immediately follows that the network coordination and anti-coordination games on a graph $\mc G$ are pairwise separable.

In both cases $\xi=\pm\1$, this is a potential game with potential 
	\be\label{potentialCA}
	\phi(x) = \pm\frac{1}{2} \sum_{i,j \in \mc V} W_{ij} x_ix_j \,.
	\ee
	Notice that for the pure network coordination game ($\xi=+\1$) the consensus profiles $\pm\1$ are always global maximum points of the potential \eqref{potentialCA} (with $+$ sign) on $\mc X$, and hence are Nash equilibria. In fact, there might be other Nash equilibria, as is known \cite{Morris:2000} that every $x$ in $\mc X$ such that $\mc V_+=\{i:\,x_i=+1\}$ and $\mc V_-=\{i:\,x_i=-1\}$ are both cohesive (c.f.~Definition \ref{definition:cohesive}) is a Nash equilibrium for the pure network coordination game. 
	
	On the other hand, the pure network anti-coordination game ($\xi=-\1$) is also a potential game, hence it admits at least two Nash equilibria $\pm x^*$ that are the global maximum points of the potential \eqref{potentialCA} (with $-$ sign) on $\mc X$.

   In contrast, when $\xi\ne\pm\1$ is not a constant vector, we obtain a mixed network coordination/anti-coordination game (cf.~\cite{Ramazi.Riehl.Cao:2016,Vanelli:2019,Vanelli.ea:2020}). Notice that if the graph $\mc G$ is connected the mixed coordination/anti-coordination game is never a potential game \cite{Arditti.Como.Fagnani:2020}. In particular, for $|\mc V|=2$ this reduces to the well known discoordination game \cite[Example 5]{Arditti.Como.Fagnani:2020}, which does not admit any pure strategy Nash equilibria. In Corollary \ref{coro:coordination-anticoordination} we shall prove that a sufficient condition for the existence of Nash equilibria of the mixed network coordination/anti-coordination game is that the set $\mc V_c=\{i\in\mc V:\,\xi_i=+1\}$ is cohesive in $\mc G$. 
	
\end{example}

%
%
%

\section{Margin of robustness of Nash equilibria}\label{sec:margin-robustness}
In this section, we introduce the notion of margin of robustness for pure strategy Nash equilibria of finite games, defined as the infimum magnitude of a perturbation that makes the configuration lose the Nash equilibrium property. We will then compute the margin of robustness of the Nash equilibria in some examples, showing in particular how the consensus configurations are the most robust Nash equilibria for a pure network coordination game. We will also show how, for potential games, the margin of robustness of Nash equilibria is not aligned with the value of the potential function, as the global maximum points of the latter may in fact be the least robust Nash equilibrium of the game. 

We start with the following simple result that will prove very useful in the rest of the paper. 
\begin{lemma}\label{lemma:basic}Let $u$ in $\mc U$ be a finite game and $x^*$ in $\mc X$ a Nash equilibrium. 
Then, for every perturbation $\delta$ in $\mc U$ such that 
\be\label{delta<-chi}||\delta_i||_{\infty}\le \frac{1}{2} \chi^{u}_i(x^*)\,,\qquad\forall i\in\mc V\,,\ee
$x^*$ is a Nash equilibrium for the perturbed game $\tilde u=u+\delta$. 
\end{lemma}
\proof
For every perturbation $\delta$ in $\mc U$ satisfying \eqref{delta<-chi}, 
 we have, 
$$\delta_i(y)-\delta_i(x^*)\le2||\delta_i||_{\infty}\le\ds u_i(x^*)-u_i(y)\,,$$
for every player $i$ in $\mc V$ and strategy profile $y\ne x^*$ such that $y\sim_i x^*$.
The above implies that
$$\tilde u_i(y)=u_i(y)+\delta_i(y)\le u_i(x^*)+\delta_i(x^*)=\tilde u_i(x^*)\,,
$$
for every player $i$ in $\mc V$ and strategy profile $y\sim_i x^*$, thus showing that $x^*$ is a Nash equilibrium for the perturbed game $\tilde u=u+\delta$. 
\qed

We now formalize the notion of margin of robustness for a Nash equilibrium of a finite game. 

\begin{definition} The \emph{margin of robustness} $\mu_u(x^*)$ of a Nash equilibrium $x^*$ in finite game $u$ in $\mc U$ is the infimum of the magnitude $||\delta||_{\infty}$ of perturbations $\delta$ in $\mc U$ such that $x^*$ is not a Nash equilibrium of the perturbed game $\tilde u=u+\delta$.
\end{definition}

The following result provides an explicit characterization of the margin of robustness for Nash equilibria of a finite game, as defined above. 

\begin{proposition}[Margin of robustness of Nash equilibria]\label{prop:margin-of-robustness}
Let $u$ be a finite game and let $x^*$ in $\mc X$ be a Nash equilibrium of $u$. 
Then, the margin of robustness of $x^*$ is 
\be\label{eq:margin-robustness}
\mu_u(x^*)=\frac12\min_{i\in\mc V}\chi^{u}_i(x^*)\,.\ee
\end{proposition}
\proof On the one hand, Lemma \ref{lemma:basic} implies  that 
\be\label{mu>=}\mu_u(x^*)\ge\frac12\min_{i\in\mc V}\chi^{u}_i(x^*)\,.\ee
On the other hand, let $i$ in $\mc V$ achieve the minimum in the righthand side of \eqref{eq:margin-robustness}, and let $y\sim_i x^*$, $y\ne x^*$ achieve the minimum in the righthand side of \eqref{chi-def} with $x=x^*$. Then, for an arbitrarily small $\eps>0$, let $\delta$ be a perturbation in $\mc U$ such that 
$$\delta_i(y)=-\delta_i(x^*)=\frac12(u_i(x^*)-u_i(y))+\eps\,,$$ 
$\delta_i(x)=0$ for every $x$ in $\mc X\setminus\{y,x^*\}$, and $\delta_j(x)=0$ for every $j$ in $\mc V\setminus\{i\}$ and $x$ in $\mc X$. 
Then, 
$$\delta_i(y)-\delta_i(x^*)=u_i(x^*)-u_i(y)+2\eps$$ 
so that 
$$\tilde u_i(y)
=u_i(y)+\delta_i(y)
= u_i(x^*)+\delta_i(x^*)+2\eps
>\tilde u_i(x^*)\,,
$$
thus showing that $x^*$ is not a Nash equilibrium of the perturbed game $\tilde u= u+\delta$. 
Observe that this perturbation has infinity norm 
$$||\delta||_{\infty}=\frac12\min_{i\in\mc V}\chi^{u}_i(x^*)+\eps\,,$$
so that by the arbitrariness of $\eps>0$ we get 
\be\label{mu<=}\mu_u(x^*)\le\frac12\chi^{u}_i(x^*)\,.\ee

The claim then follows by combining \eqref{mu>=} and \eqref{mu<=}. 
\qed

In the remaining of this section we provide various examples characterizing the robustness of relevant games.


\begin{example}\label{example:prisoner}
 \begin{figure}
	 	\centering
	 	\begin{tabular}{|c|c|c|}
	 		\hline 
	 		& -1 & +1 \\ 
	 		\hline 
	 		-1 & a,a & c,d \\ 
	 		\hline 
	 		+1 & d,c & b,b \\ 
	 		\hline 
	 	\end{tabular} 
	 	\caption{Normal form representation of the prisoner dilemma game.}
	 	\label{fig:table}	 
	 \end{figure}
Consider a symmetric $2$-player binary action game $u$ with utilities as in Figure \ref{fig:table}, where
\be\label{prisoner-inequalities}c>b>a>d\,.\ee
This is the classical Prisoner Dilemma game whereby action $-1$ is to be interpreted as ``Defect''  and action $+1$ as ``Cooperate''. 
As known, $-1$ is a  strictly dominant action for both players in this game, so that there is a unique pure strategy Nash equilibrium $x^*=(-1,-1)$. For both players $\{1,2\}$ we have that 
$\chi^{u}_1(x^*)= \chi^{u}_2(x^*)=a-d$
so that the margin of stability  is
$$\mu_u(x^*) = a-d \,,$$
i.e., the difference between the utility that players get when they both defect minus the one that a player gets when she cooperates and the other one defects. 
\end{example}

\begin{example}[One shot public good  game]
\label{example:publicgood} 
For an undirected graph $\mc G=(\mc V, \mc E)$ and a scalar value $c$ such that $0<c<1$, consider the game with player set $\mc V$ coinciding with the node set of $\mc G$
and utilities 
$$
\ba{rclcl}
u_i(x)&=& 1-c &\se&  x_i=1\\
u_i(x)&=& 1 &\se& x_i=0, x_j=1 \text{ for some } j \in \mc N_i\\
u_i(x)&=& 0 &\se& x_i=0, x_j=0 \text{ for all } j \in \mc N_i\,.
\ea$$
This is known as the one shot public good game \cite{Jackson.Zenou:2015}. Its pure strategy Nash equilibria are all those strategy profiles $x^*$ in $\{0,1\}^{\mc V}$ such that $\{i\in\mc V:\,x_i=1\}$ is a maximal independent set of $\mc G$. For every such Nash equilibrium $x^*$, it is easily computed that 
\begin{equation*}
\chi^{u}_i(x^*)=\begin{cases}
1-(1-c) = c \quad \text{if } x^*_i=0 \\
1-c \quad \text{if } x^*_i=1\,,
\end{cases}
\end{equation*}
so that the margin of robustness is 
\begin{equation*}
\begin{aligned}
\mu_u(x^*)&=\frac12\min_{i\in\mc V}\chi^{u}_i(x^*)\\
&= \frac12\min\{c, (1-c)\}\,.
\end{aligned}
\end{equation*} 
In particular, this means that all Nash equilibria of the one shot public good game have the same margin of robustness.
\end{example}


\begin{example}[Robustness of network coordination games]\label{example:robustnessCoord}
For a weighted undirected graph $\mc G=(\mc V,\mc E,W)$, consider the pure network coordination game with utilities \eqref{coord-anticoord-def} with $\xi_i=+1$ for every player $i$ in $\mc V$.   
For this game, as observed in Example \ref{example:coordination-anticoordination}, the consensus configurations $x^*=\pm\1$ are both Nash equilibria. 
In either such configuration, if a player $i$ changes her action from $x^*_i$ to $-x^*_i$, her utility decreases by twice her degree. This implies that 
\begin{equation}\label{eq:chiCoord}	\chi_i^u(x^*) = u_i(x^*_i,x^*_{-i}) - u_i(-x^*_i,x^*_{-i}) = 2w_i \,.
\end{equation}
	Then,  Proposition \ref{prop:margin-of-robustness} implies that the margin of robustness for $x^*=\pm \1$ coincides with the minimum degree
\be\label{margin-consensus}	\mu_u(\pm\1) = \min_{i \in \mc V}w_i \,.\ee
In fact, in every other Nash equilibrium $x^*$ of the pure network coordination game, it is easily seen that if a player $i$ changes her action from $x^*_i$ to $-x^*_i$, her utility decreases by no more than twice her degree. 
$$	\chi_i^u(x^*) = u_i(x^*_i,x^*_{-i}) - u_i(-x^*_i,x^*_{-i}) 
	\leq 2w_i\,, $$
	so that 
\be\label{margin-nonconsensus}	\mu_u(x^*) \leq \min_{i \in \mc V}w_i \qquad \forall\text{ Nash equilibrium }x^*\,.\ee
It follows from \eqref{margin-consensus} and \eqref{margin-nonconsensus}, that the maximally robust Nash equilibra for the pure network coordination game coincide with the the consensus configurations, which are also the maximum points of the potential. As we shall see in the following examples, this far from being true in general. 
\end{example}

	
	\begin{example}[Robustness of network anti-coordination]\label{example:robustnessAntiCoord}
		\begin{figure}
			\centering
			\includegraphics[width=0.2\textwidth]{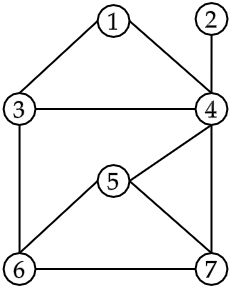}
			\caption{Graph for the network anti-coordination game of Example \ref{example:robustnessAntiCoord}.}
			\label{fig:graphAnti}
		\end{figure}
		Let $\mc G$ be the simple graph in Figure \ref{fig:graphAnti} and let $u$ be the anti-coordination game on $\mc G$, with utilities given by \ref{coord-anticoord-def} with $\xi_i=-1$ for every player $i$ in $\mc V$. 

	
	It is easily verified that the potential \eqref{potentialCA} (with $-$ sign) achieves its maximum value $\phi(x^*) = 6$ in the strategy profiles  
	\be\label{maximum-potential}x^*=\pm(1,-1,-1,1,-1,1,-1)\,,\ee  that are illustrated in Figure \ref{fig:Nash-anticoord} (a). Hence the strategy profiles in \eqref{maximum-potential} are Nash equilibria. For them we have that 
	\begin{equation*}
	\begin{aligned}
	&\chi_1^u(x^*)=0, \chi_2^u(x^*)=1, \chi_3^u(x^*)=3,\chi_4^u(x^*)=3,\\
	& \chi_5^u(x^*)=1, \chi_6^u(x^*)=3, \chi_7^u(x^*)=1 \,,
	\end{aligned}
	\end{equation*}
	so that their margin of robustness is $\mu_u(x^*)=0$. Indeed, they are non-strict Nash equilibria as player $1$ might switch her action without changing her utility.
	
In contrast, consider the strategy profiles 
\be\label{most-robust}	x^*=\pm(1,1,-1,-1,-1,1,1)\,,\ee
which are illustrated in Figure \ref{fig:Nash-anticoord} (b).
\begin{figure}
     \centering
     \subfigure[Nash equilibrium \eqref{maximum-potential}.]{\includegraphics[width=0.2\textwidth]{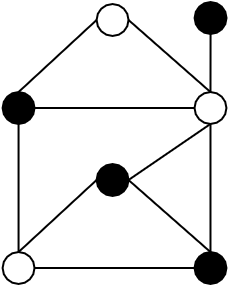}}
     \hfill
      \subfigure[Nash equilibrium \eqref{most-robust}.]{\includegraphics[width=0.2\textwidth]{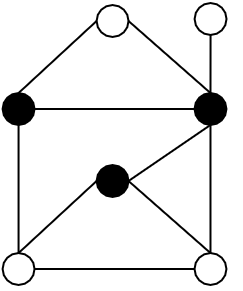}}
     \caption{\label{fig:Nash-anticoord}Two Nash equilibria for the anti-coordination game of Example \ref{example:robustnessAntiCoord}. The Nash equilibrium in (a) achieves the maximum value  $\phi(x^*)=6$ of the potential function  and has margin of robustness $\mu_u(x^*)=0$. The Nash equilibrium in (a) achieves a lower value $\phi(x^*)=4$ of the potential function and has margin of robustness $\mu_u(x^*)=1$.  }
\end{figure}

For $x^*$ as in \eqref{most-robust}, we can compute $\phi(x^*)=4$ and 
	\begin{equation*}
	\begin{aligned}
	&\chi_1^u(x^*)=2, \chi_2^u(x^*)=1, \chi_3^u(x^*)=1,\chi_4^u(x^*)=1,\\
	& \chi_5^u(x^*)=1, \chi_6^u(x^*)=1, \chi_7^u(x^*)=1 \,,
	\end{aligned}
	\end{equation*}
	showing that these are Nash equilibria with margin of robustness $\mu_u(x^*)=1$. 
	\end{example}
Observe that Example \ref{example:robustnessAntiCoord} illustrates how the margin of robustness of Nash equilibria in potential games is generally not aligned with the potential function, as the maximally robust Nash equilibria might not be the maximum points of the potential function. In contrast, in Example \ref{example:robustnessCoord} we showed that the maximally robust Nash equilibria of a pure network coordination game coincide with the global maximum points of its potential. 
In fact, one may consider the pure network coordination game an exception rather than the norm.

\section{Robustness of equilibria in subgames}\label{sec:robustness-subgames}


In this section, we focus perturbations having a certain structure emerging from considering the original game as a sub-game of another one with a larger set of players. Particularly, we investigate conditions guaranteeing the existence of a Nash equilibrium that does not depend on the behavior of the players not in the original game. As we shall see in the following section, the results obtained here will prove particularly useful in the analysis of network games. 

We start by introducing some further notation. 
For a nonempty subset of players $\mc R\subseteq\mc V$, let
$\mc S=\mc V\setminus\mc R\,,$
be the complementary set of players. 
We shall consider the space of games $\mc U_{\mc R}$ with player set $\mc R$ and strategy profile space $\mc X_{\mc R}=\prodd\nolimits_{i\in\mc R}\mc A_i$. We can canonically identify $\mc X$ with $\mc X_{\mc R}\times \mc X_{\mc S}$ and consequently decompose every $x$ in $\mc X$, as $x=(x_{\mc R}, x_{\mc S})$
with $x_{\mc R}$ in $\mc X_{\mc R}$ and  $x_{\mc S}$ in $\mc X_{\mc S}$.
Observe that $\mc U_{\mc R}$ can be interpreted as a subspace of $\mc U$ and, with a slight abuse of notation, we shall define the distance between a game $u$ in $\mc U$ and a game $\tilde u$ in $\mc U_{\mc R}$ as 
$$||u_i-\tilde u_i||_{\infty}=\max_{x\in\mc X}|u_i(x)-\tilde u_i(x_{\mc R})|\,,\qquad \forall i\in\mc R\,,$$
$$||u-\tilde u||_{\infty}=\max_{i\in\mc R}||u_i-\tilde u_i||_{\infty}\,.$$

For a given game $u$ in $\mc U$, we shall consider the following games in $\mc U_{\mc R}$: for every $z$ in $\mc X_{\mc S}$ let the game $u^{(z)}$ in $\mc U_{\mc R}$ have utilities 
\be\label{uz-def}u_i^{(z)}(y)=u_i(y,z)\,,\qquad \forall y\in\mc X_{\mc R}\,,\ee
for every player $i$ in $\mc R$. 
It will also prove useful to consider the following averaged game
$\ov u^{\mc R}$ in $\mc U_{\mc R}$ having utilities 
\be\label{ovu-def}\bar u^{\mc R}_i(y)=\frac1{|\mc X_{\mc S}|}\sum_{z\in\mc X_{\mc S}}u_i^{(z)}(y)\,,\qquad\forall y\in\mc X_{\mc R}\,,\ee
for every player $i$ in $\mc R$. Games $u^{(z)}$ and $\bar u^{\mc R}$ will be called $\mc R$-restricted games.

We now study conditions under which there exists a configuration in $\mc X_{\mc R}$ that is a Nash equilibrium for all games $u^{(z)}$. 
We state the following result.

\begin{proposition}\label{prop:uniform-Nash} 
For a game $u$ in $\mc U$ and a nonempty subset of players $\mc R\subseteq\mc V$, 
let $y^*$ in $\mc X_{\mc R}$ be a Nash equilibrium of the game $\ov u$ in $\mc U_{\mc R}$ defined in \eqref{ovu-def}. If
\be\label{sufficient}\chi^{\ov u}_i(y^*)\ge2||u_i-\ov u_i||_{\infty}\,,\qquad \forall i\in\mc R\,,\ee
then  $y^*$ is a Nash equilibrium of the game $u^{(z)}$ in $\mc U_{\mc R}$, for every $z$ in $\mc X_{\mc S}$.
\end{proposition}
\proof 
For every $z$ in $\mc X_{\mc S}$, define the perturbation
$$\delta^{(z)}=u^{(z)}-\ov u\in\mc U_{\mc R}\,.$$
Then, observe that
$$||u_i-\ov u_i||_{\infty}=\max_{z\in\mc X_{\mc S}}||\delta^{(z)}_i||_{\infty}\,,$$
for every player $i$ in $\mc R$,
so that \eqref{sufficient} implies that 
$$||\delta^{(z)}_i||_{\infty}\le\frac12\chi^{\ov u}_i(y^*)$$
for every $z$ in $\mc X_{\mc S}$. It then follows from Lemma \ref{lemma:basic} that $y^*$ is a Nash equilibrium of the game $u^{(z)}$ for every $z$ in $\mc X_{\mc S}$.
\qed

Proposition \ref{prop:uniform-Nash} provides sufficient conditions for the existence of a uniform Nash equilibrium $y^*$ in $\mc X_{\mc R}$ of the game obtained by restricting $u$ to a subset of players $\mc R$; under such conditions $y^*$ is a Nash equilibrium in all games 
$u^{(z)}$ in $\mc U_{\mc R}$ obtained by freezing the strategies of all players in the complementary set $\mc S=\mc V\setminus\mc R$. In the next section, it will become apparent how this result may prove very useful in the analysis of network games.

We conclude this section with the following result that may be thought of as a  sort of converse of Proposition \ref{prop:uniform-Nash}.
\begin{proposition}\label{coro:uovu1}
Let $u$ in $\mc U$ be a finite game and $x^*$ in $\mc X$ a Nash equilibrium for $u$. Then, for every nonempty subset of players $\mc R\subseteq\mc V$, 
the projection $x^*_{\mc R}$ is a Nash equilibrium for every game $\tilde u$ in $\mc U_{\mc R}$, such that
\be\label{sufficient-2}||u-\tilde u||_{\infty}\le\mu_u(x^*)\,.\ee 
\end{proposition}
\proof 
For every $z$ in $\mc X_{\mc S}$, define the games $u^{(z)}$ and $\delta^{(z)}$ in $\mc U_{\mc R}$ with utilities
$$u_i^{(z)}(y)=u_i(y,z)\,,\qquad\delta^{(z)}_i(y)=\tilde u_i(y)-u_i^{(z)}(y)\,,$$  
for every player $i$ in $\mc R$ and strategy profile $y$ in $\mc X_{\mc R}$. 
Now,  let $y^*=x^*_{\mc R}$ and $z^*=x^*_{\mc S}$. Then, clearly $y^*$ is a Nash equilibrium of $u^{(z^*)}$. 
Moreover, 
\eqref{sufficient-2} implies that 
$$||\delta^{(z^*)}||_{\infty}\le\mu_{u^{(z^*)}}(y^*)\,,$$
so that, by Proposition \ref{prop:margin-of-robustness} $y^*$ is a Nash equilibrium of $\tilde u$. 
 \qed

%
%
%
%

\section{Robustness of equilibria in network games}\label{sec:robustness-networkgames}
This section focuses on pairwise network games $u$ as per Definition \ref{definition:pairwise-network-game}.  
More specifically, for a binary partition $\mc V=\mc R\cup\mc S$ of the player set, we show how the existence of a Nash equilibrium $x^*$  of the network game $u$ can be guaranteed by sufficient conditions that involve the robustness and structure of the sub-games in $\mc U_{\mc R}$ and $\mc U_{\mc S}$, respectively,  as well as on the strength of the network coupling between the sets $\mc R$ and $\mc S$. Our general results will find a direct application  for the mixed network coordination/anti-coordination games introduced in Example \ref{example:coordination-anticoordination}. 

Consider a directed graph $\mc G=(\mc V, \mc E)$ and a pairwise-separable $\mc G$-game $u$ in $\mc U$ with utilities as in \eqref{pairwise-game} for given pairwise utilities $u_{ij}:\mc A_i\times\mc A_j\to\R$ for every $(i,j)$ in $\mc E$. First, the following result follows from Proposition \ref{prop:uniform-Nash}. 
\begin{proposition}\label{prop:coupling}
Let $\mc R\subseteq\mc V$ be a nonempty subset of players and $\mc S=\mc V\setminus \mc R$.
Let $\ov u$ in $\mc U_{\mc R}$ be the game defined in \eqref{ovu-def} and $y^*\in\mc X_{\mc R}$ be a Nash equilibrium for $\ov u$.  If 
\be\label{sufficient-2}\chi^{\ov u}_i(y^*)\geq M w_i^{\mc S}\,,\ee
for every $i$ in $\mc R$, where 
$$M=\max_{i\in\mc R, j\in\mc S}||u_{ij}||_{\infty}\,,\qquad w_i^{\mc S}=|\mc N_i\cap\mc S|\,,$$
then  $y^*$ is a Nash equilibrium of the game $u^{(z)}$ in $\mc U_{\mc R}$, for every $z$ in $\mc X_{\mc S}$.
\end{proposition}
\proof 
For every $i$ in $\mc R$ we have
$$\begin{array}{rcl}||u_i-\ov u_i||_{\infty}\!\!\!
&=&\max\limits_{x\in\mc X}\frac{1}{|\mc X_{\mc S}|}\sum\limits_{z}\left(u_i(x_{\mc R}, x_{\mc S})-u_i(x_{\mc R}, z)\right)\\[3pt] 
&=&\max\limits_{x\in\mc X}\frac{1}{|\mc X_{\mc S}|}\sum\limits_{z}\sum\limits_{j}\left(u_{ij}(x_i, x_j)-u_{ij}(x_i, z_j)\right)\\[3pt] 
&\leq& 2Mw_i^{\mc S}\,,\end{array}$$
where the summation indices $z$ and $j$ run over $\mc X_{\mc S}$ and $N_i\cap\mc S$, respectively. 
Hence,  \eqref{sufficient} holds true and the claim then follows from Proposition \ref{prop:uniform-Nash}. 
\qed 

Proposition \ref{prop:coupling} readily implies the following: 
\begin{theorem}\label{theo:Nashcoupled} 
For a pairwise separable network game $u$ on a graph $\mc G=(\mc V,\mc E)$, 
and a nonempty subset of players  $\mc R\subseteq\mc V$, let $\ov u$ in $\mc U_{\mc R}$ be the game defined in \eqref{ovu-def} and $\mc S=\mc V\setminus \mc R$. 
If 
\begin{enumerate}
\item[(a)] $\bar u$ admits a Nash equilibrium $y^*$ in $\mc X_{\mc R}$ satisfying \eqref{sufficient-2};
\end{enumerate}
and
\begin{enumerate}
\item[(b)] for every $y\in\mc X_{\mc R}$ the game $u^{(y)}$ in $\mc U_{\mc S}$ admits Nash equilibria $z^*(y)$;
\end{enumerate}
then $x^*=(y^*,z^*(y^*))$ in $\mc X$ is a pure strategy Nash equilibrium of $u$.
\end{theorem}
\begin{proof}
It follows from assumption (a) and Proposition \ref{prop:coupling} that $y^*$ is a Nash equilibrium of the game $u^{(z)}$ in $\mc U_{\mc R}$ for every $z\in\mc X_{\mc S}$. 
Hence, in particular, $y^*$ is a Nash equilibrium of the game $u^{(z^*(y^*))}$, while assumption (b) guarantees that $z^*(y^*)$ is Nash equilibrium of $u^{(y^*)}$. It then follows that $x^*=(y^*,z^*(y^*))$  is a Nash equilibrium of  $u$.
\end{proof}

In order to apply Theorem \ref{theo:Nashcoupled} in concrete examples, the following result proves useful in order to guarantee that assumption (b) holds true. 
\begin{proposition}\label{proposition:pairwisepotential} Assume that $\mc G=(\mc V, \mc E)$ is undirected and for every $(i,j)$ in $\mc S\times\mc S$, the pairwise two-player game between $i$ and $j$ with utilities $u_{ij}(x_i,x_j)$ and $u_{ji}(x_j,x_i)$ is a symmetric potential game. Then,  for every $y\in\mc X_{\mc R}$, $u^{(y)}$ in $\mc U_{\mc S}$ is a potential game, hence it admits at least one Nash equilibrium $z^*(y)$ in $\mc X_{\mc S}$.  
\end{proposition}
\begin{proof} For every undirected link $\{i,j\}$, let $\phi_{ij}(x_i,x_j)=\phi_{ji}(x_i,x_j)$ be a potential function of the two-player game between $i$ and $j$ with utilities $u_{ij}(x_i,x_j)$ and $u_{ji}(x_j,x_i)$ and let 
$$\phi(x)=\frac{1}{2}\sum_{\substack{i,j\in\mc S:\\(i,j)\in\mc E}}\phi_{ij}(x_i,x_j)+\sum_{\substack{i\in\mc S, j\in\mc R:\\(i,j)\in\mc E}}u_{ij}(x_i,x_j)\,.$$
Then, $\phi(x)$ is a potential function of the game $u^{(y)}$. 
\end{proof}
 
 We are now ready to apply Theorem \ref{theo:Nashcoupled} to the mixed network coordination anti-coordination game introduced in Example \ref{example:coordination-anticoordination}. We do that after introducing the following definition of cohesiveness, first introduced in \cite{Morris:2000}. 
 
 \begin{definition}[Cohesiveness]\label{definition:cohesive}
In a graph $\mc G=(\mc V,\mc E)$, a subset of nodes $\mc R\subseteq\mc V$ is \emph{cohesive} if 
$$w_i^{\mc R}\ge w_i^{\mc S}\,,\qquad \mc S=\mc V\setminus\mc R\,,$$
for every node $i$ in $\mc R$,  
where $$w_i^{\mc R}=|\mc N_i\cap\mc R|\,,\qquad w^{\mc S}_i=|\mc N_i\cap\mc S|\,.$$
\end{definition}\medskip

The following result first appeared in the master thesis \cite{Vanelli:2019}. 

\begin{corollary}\label{coro:coordination-anticoordination}
Let $\mc G=(\mc V,\mc E)$ be an undirected graph and let $\mc V_c\subseteq\mc V$, $\mc V_a=\mc V\setminus\mc V_c$ be subsets of players. 
Let $u$ in $\mc U$ be the mixed network coordination/anti-coordination game with set of coordinating players $\mc V_c$ and set of anti-coordinationg players $\mc V_a$. 
If $\mc V_{c}$ is cohesive in $\mc G$, then $u$ admits a pure strategy Nash equilibrium $y^*$ (in fact at least two). 
\end{corollary}
\proof 
Let $\ov u$  and  $u^{(z)}$ for $z$ in $\mc X_{\mc V_a}$, be the games in $\mc U_{\mc V_c}$ defined as in \eqref{ovu-def} and \eqref{uz-def}, respectively, with  
$\mc R=\mc V_c$ and $\mc S=\mc V_a$. Observe that $\ov u$ coincides with the coordination game on the induced graph $\mc G_{\mc V_c}$.  
Hence, any consensus vector $y=\pm\1$ in $\mc X_{\mc V_c}$ is a Nash equilibrium of the game $\ov u$ and \eqref{eq:chiCoord} implies that 
\be\label{chicoordination}\chi^{\ov u}_i(y^*)=2w_i^{\mc R}\,,\qquad\forall i\in\mc R\,.\ee
Since  $||u_{ij}||_{\infty}\le M=1$ for every $(i,j)$ in $\mc E$, and the set $\mc V_c$ is cohesive in $\mc G$,  \eqref{chicoordination} implies that 
$$\chi^{\ov u}_i(y^*)=2w_i^{\mc R}\ge2 Mw_i^{\mc S},\qquad\forall i\in\mc R\,,$$
so that \eqref{sufficient-2} holds true. Then, assumption (a) of Theorem \ref{theo:Nashcoupled} is satisfied by any consensus vector $y^*=\pm\1$ in $\mc X_{\mc R}$. 
 
On the other hand, since a two-player anti-coordination game is a potential game and the graph $\mc G$ is undirected, by applying Proposition \ref{proposition:pairwisepotential} with $\mc S=\mc V_a$ we get that for every $y$ in $\mc X_{\mc V_c}$, the game $u^{(y)}$ in $\mc U_{\mc V_a}$ is a potential game, hence it admits at least one Nash equilibrium $z^*(y)$ in $\mc X_{\mc V_a}$. Then, assumption (b) of Theorem \ref{theo:Nashcoupled} is satisfied. 

Theorem \ref{theo:Nashcoupled} then implies that, for every consensus vector $y^*=\pm\1$ in $\mc X_{\mc V_c}$, the strategy profile $(y^*,z^*(y^*))$ in $\mc X$ is a Nash equilibrium of the mixed network coordination/anti-coordination game $u$, thus proving the claim.  
\qed

Corollary \ref{coro:coordination-anticoordination} provides a sufficient  graph-theoretic condition for the existence of pure strategy Nash equilibria in mixed network coordination-anticoordination games: cohesiveness of the set of coordinating agents. As observed, existence of pure strategy Nash equilibria for mixed network coordination-anticoordination games cannot be guaranteed in general. E.g., in the special case of a simple graph with two nodes connected by an undirected link, where the first node is coordinating and the other one is anti-coordinating, this reduces to the discoordination game which is well-known not to admit pure strategy Nash equilibria. Corollary \ref{coro:coordination-anticoordination} significantly generalizes previous works where existence of pure strategy Nash equilibria was proved only for pure coordination or anti-coordination games \cite{Ramazi.Riehl.Cao:2016}.

We conclude this section with following example illustrating concretely how Corollary \ref{coro:coordination-anticoordination} can be applied.
\begin{example}\label{example:mixedNash}
	\begin{figure}
		\centering
		\includegraphics[width=0.5\linewidth]{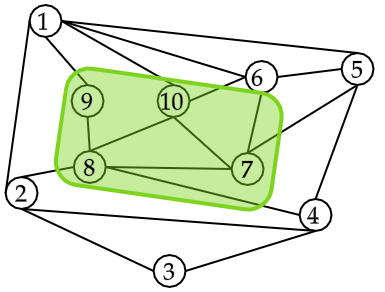}
		\caption{Graph $\mc G$ for Example \ref{example:mixedNash}. The set $\mc V_c$ is marked in green. }
		\label{fig:mixed}
	\end{figure}
	
	Consider the mixed coordination/anti-coordination game $u$ on the graph $\mc G=(\mc V,\mc E)$ represented in Figure \ref{fig:mixed}, with set of coordinating players $\mc V_c = \{7,8,9,10\}$ and set of anti-coordinating players $\mc V_a = \mc V \setminus \mc V_c$. $\mc V_c$ is a cohesive set in $\mc G$, so we can construct a Nash equilibrium for $u$ by combining a consensus configuration $y^*$ in $ \mc X_{\mc V_c}$ for players in $\mc V_c$ with a configuration $z^*(y^*)$ of the remaining players $\mc V_a$ which maximizes the potential of the resulting anti-coordination game $u^{(y^*)} $ in $\mc U_{\mc V_a}$.
	\begin{figure}
		\centering
		\includegraphics[width=0.4\linewidth]{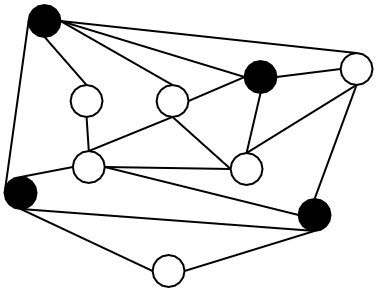}
		\caption{Representation of two Nash equilibria for the mixed game of Example \ref{example:mixedNash}.}
		\label{fig:mixednash}
	\end{figure}
	In this way we obtain the two Nash configurations represented in Figure \ref{fig:mixednash}.
\end{example}

\section{Conclusion}
In this paper, we have studied the robustness of  Nash equilibria of finite games. First we have introduced a notion of margin of robustness of a pure strategy Nash equilibrium of a game and shown how it can be computed explicitly by a simple formula. Then, we have refined this notion in order to study subgames suitably defined by restricting the player set of game to a subset. Finally, we have shown how these results can be applied to network games, in particular proving some graph-theoretical sufficient conditions for the existence of a pure strategy Nash equilibria for mixed coordination/anti-coordination games.  

We consider these results as a promising preliminary study that is worth being extended in several directions, including: (i) solving for the maximally robust Nash equilibrium in finite games with multiple equilibria; (ii) generalizing the notion of margin of robustness by considering different norms for the game perturbations; (iii) extending the notion of robustness of Nash equilibria to learning dynamics for games, e.g., best response and noisy best response dynamics. 
%
%
%
%

%
%

\bibliographystyle{unsrt}
\bibliography{bib}

\end{document}